\documentclass[11pt]{article}
\usepackage[letterpaper,margin=1in]{geometry}
\usepackage[utf8]{inputenc}
\usepackage[T1]{fontenc}
\usepackage[UKenglish]{babel}

\usepackage{amsmath,amssymb,amsthm,mathtools}
\usepackage{microtype}
\usepackage{booktabs}
\usepackage{enumitem}
\usepackage{orcidlink}
\usepackage{xcolor}
\usepackage{graphicx}
\usepackage{tikz}
\usepackage{pgfplots}
\pgfplotsset{compat=1.18}

\usepackage{hyperref}
\usepackage[capitalise,noabbrev]{cleveref}

\hypersetup{
  colorlinks=true,
  linkcolor=blue,
  citecolor=blue,
  urlcolor=blue,
  pdfauthor={},
  pdftitle={The Completion-Threshold Framework for Obligatory-Test Scheduling on Multiple Machines},
  pdfsubject={},
  pdfkeywords={}
}

\newtheorem{theorem}{Theorem}[section]
\newtheorem{lemma}[theorem]{Lemma}
\newtheorem{corollary}[theorem]{Corollary}
\newtheorem{definition}[theorem]{Definition}

\newcommand{\OPT}{\mathrm{OPT}}
\newcommand{\ALG}{\mathrm{ALG}}

\graphicspath{{./img/}}

\title{Hardness of Obligatory-Test Scheduling on Multiple machines}
\begin{document}

\begin{titlepage}
\thispagestyle{empty}

\begin{center}
{\LARGE\bfseries
Hardness of Obligatory-Test Scheduling on\\
Multiple Machines
\par}

\vspace{1.5em}

{\large
Kao-Chuan Liang\footnote{\texttt{daniel.cs12@gapp.nthu.edu.tw}}
\qquad
Ya-Chun Liang\footnote{\texttt{ycliang@cs.nthu.edu.tw} } 
% \orcidlink{0000-0002-7359-4335}
\par} 

\vspace{0.6em}

{\normalsize
Department of Computer Science\\
National Tsing Hua University\\
Hsinchu, Taiwan
\par}

% \vspace{0.6em}
% {\small Ya-Chun Liang: \href{https://orcid.org/0000-0002-7359-4335} {\texttt{0000-0002-7359-4335}} \par}
\end{center}

\vspace{1.5em}
\begin{abstract}
We study online scheduling with obligatory testing on $m$ identical parallel machines, with the objective of minimizing the sum of completion times. 
%%yc
Each job comprises a test of known length and a processing operation of initially unknown length. The processing time is revealed only when the test completes.
% In this model, each job must first undergo a test, and its processing time is revealed only after the test completes. 
%%yc
Unlike in optional testing models, the scheduler does not choose whether to acquire information. Instead, it must decide how to allocate machine capacity between testing unrevealed jobs and processing jobs whose sizes are already known.
% Thus the algorithmic challenge is not whether to acquire information, but how to allocate parallel machine capacity between revealing unknown jobs and processing already revealed jobs. 
%%yc
Previous single-machine lower-bound constructions suggest a natural 
%%yc
%%kc
$\sqrt{2}$ benchmark
% ~\cite{DogeasErlebachLiang2024} 
[ESA 2024: 48:1-14].
% baseline $\sqrt{2}$, 
%%yc 
However, these constructions 
%%yc
cannot be directly transferred
% do not transfer formally
to identical parallel machines by a simple replication argument. An online algorithm may interleave jobs from different copies, and the test and processing operation of a job need not be scheduled on the same machine.
% but they do not formally transfer to identical parallel machines by a simple copy argument: an online algorithm may interleave jobs from different copies and may process operations belonging to the same copy on different machines.
We address this difficulty by introducing a completion-threshold framework that reasons directly about global progress under total machine 
capacity. For each $X$, 
%%yc
let $T_X$ be
% the threshold $T_X$ is 
the earliest time at which the algorithm has completed at least $X$ jobs.
The identity $\sum_{X=1}^{N}T_X$
%%yc
then converts
% turns
pointwise progress bounds into lower bounds on the total completion time.
Using this framework, we prove a three-type lower bound of $1.4811$ and a dyadic multi-type lower bound 
%%yc
tending
% that tends
to~$3/2$. The latter also improves the deterministic single-machine lower bound from $\sqrt{2}$ to~$3/2$. On the algorithmic side, we give a parallel version of single-machine 1-SORT and prove that, if single-machine 1-SORT is $\rho$-competitive, then its parallel version is $\frac{2(m+\rho-1)}{m+1}$-competitive on $m$ identical machines.
\end{abstract}

\end{titlepage}

\setcounter{page}{1}

\section{Introduction}
The framework of scheduling with testing 
%%yc
captures
% involves
problems in which the processing time of a job is not known in advance, but can be revealed by performing a test. Most previous work considers the optional-test model, where the scheduler may decide whether to test a job first or to execute it directly using prior information. In contrast, the authors in~\cite{DogeasErlebachLiang2024} recently introduced the obligatory-test model, in which every job must be tested before its processing part can begin, and studied the single-machine case with the objective of minimizing the sum of completion times. In this paper, we extend the obligatory-test model to $m$ identical parallel machines. Each job $j$ consists of a test of known length $t_j$ and a processing operation of 
%%yc add
initially
unknown length $p_j$, where $p_j$ is revealed only 
%%yc
when the test completes.
% upon completion of the test.
All jobs are available at time $0$, and the objective is to minimize $\sum_j C_j$, where $C_j$ denotes the completion time of job $j$. We consider both standard test-time settings. In the uniform setting, all jobs have the same test time, and we assume without loss of generality that $t_j=1$ for all $j$. In the non-uniform setting, test times may vary across jobs.

The obligatory-test model is structurally different from the optional one. In the optional model, the scheduler must decide whether information acquisition is worthwhile. In the obligatory model, this decision disappears, since every job must eventually be tested.
The algorithm must instead decide how to divide machine time between two competing activities: testing jobs whose processing times are still unknown, and processing jobs whose test outcomes have already been revealed. If too much time is spent on testing, information is acquired quickly but a large processing backlog may accumulate. If too much time is spent on processing, the current backlog is reduced but future jobs are revealed more slowly.
This tension already appears on a single machine, but it becomes more subtle on parallel machines, where multiple tests and multiple processing operations may proceed simultaneously. Thus, extending the model to multiple identical machines is not merely a routine generalization.
It is the natural setting in which to ask whether parallelism can fundamentally mitigate the tradeoff between revealing jobs and processing revealed 
%%yc add
already
jobs.

Our goal is to understand this tradeoff in the multiple-machine obligatory setting. Single-machine lower-bound constructions provide useful intuition, but they do not directly yield a formal lower bound for identical parallel machines by simply replicating hard instances. After such a replication, an online algorithm is not forced to keep the copies separated.
It may
%%yc add
instead
interleave jobs from different copies and may process operations belonging to the same copy on different machines. Thus, a 
%%yc 
lower bound for the multi-machine setting
% genuine multi-machine lower bound 
must reason directly about the total machine capacity available over time. The single-machine exchange argument also does not directly extend. On one machine, the schedule is a single linear sequence, and an exchange argument can move the processing operation of a job close to its test while preserving a clear order of operations. On multiple machines, there is no canonical exchange.
The test and the processing operation of the same job may lie on different machine timelines, and moving one operation may change the relative order of operations on several machines at once. We therefore use completion thresholds as the main analytical object. For each $X$, let $T_X$ be the earliest time at which the algorithm has completed at least $X$ jobs. Since
\[
    \ALG=\sum_{X=1}^{N}T_X,
\]
it suffices to lower-bound the global rate at which jobs can be revealed and completed. This viewpoint is particularly suited to obligatory testing, because every completed job must first pass through a test before its processing time is known.

Using this framework, we prove lower bounds that go beyond the single-machine 
%%yc
benchmark $\sqrt{2}$ from~[5].
% baseline $\sqrt{2}$. 
A three-type construction gives an asymptotic lower bound of $1.4811$, attained around $\alpha\approx 0.1939$ and $\beta\approx 0.2873$.
A dyadic multi-type construction gives a lower bound tending to $3/2$, and also improves the deterministic single-machine lower bound from $\sqrt{2}$ to $3/2$.  
On the algorithmic side, we prove that the single-machine 1-SORT guarantee lifts to $m$ identical machines with ratio $\frac{2(m+\rho-1)}{m+1}$, where $\rho = 1.861$ for arbitrary test times and $\rho = 1.585$ for uniform test times.

\begin{table}[htbp]
\centering
\renewcommand{\arraystretch}{1.18}
\begin{tabular}{c|c|c|c}
\toprule
\textbf{Setting} & \textbf{Machines} & \textbf{Best known upper bound} & \textbf{Lower bound} \\
\midrule
\textbf{Uniform tests} & 1
& $1.585$~\cite{DogeasErlebachLiang2024}
& $\sqrt{2}$~\cite{DogeasErlebachLiang2024} $\rightarrow$ $\frac32$ \\

\textbf{Arbitrary tests} & 1
& $1.861$~\cite{DogeasErlebachLiang2024}
& $\sqrt{2}$~\cite{DogeasErlebachLiang2024} $\rightarrow$ $\frac32$ \\

\textbf{Uniform tests} & $m$
& $\frac{2(m+\rho-1)}{m+1}$
& $\frac32$ \\

\textbf{Arbitrary tests} & $m$
& $\frac{2(m+\rho-1)}{m+1}$
& $\frac32$ \\
\bottomrule
\end{tabular}
\caption{Bounds for the obligatory-test model. The arrows indicate the improvements obtained in this paper. In the multiple-machine upper bounds, \(\rho\) denotes the corresponding single-machine 1-SORT ratio: \(\rho=1.585\) for uniform tests and \(\rho=1.861\) for arbitrary tests.}
\label{fig:landscape}
\end{table}

\subparagraph{Related Work.}
The framework of scheduling with testing can be studied under both stochastic and adversarial uncertainty. In the stochastic setting, testing reveals random job characteristics, and the objective is to minimize expected costs~\cite{LeviMagnantiShaposhnik2019}. This probabilistic approach differs from the adversarial online setting, which originated with the optional-test model introduced 
%%yc delete 
% by the authors 
in~\cite{DurrErlebachMegowMeissner2018,DurrErlebachMegowMeissner2020}. In this model, jobs have known upper bounds, and the scheduler decides whether to test them 
%%yc
in order to reveal their
% to reveal 
actual processing times. For a single machine with uniform test times, the authors in~\cite{DurrErlebachMegowMeissner2020} established tight bounds for makespan minimization and provided competitive algorithms and lower bounds for minimizing the sum of completion times. This was later extended to arbitrary test times 
%%yc
in~\cite{AlbersEckl2021WAOA},
% by Albers and Eckl~\cite{AlbersEckl2021WAOA}, 
with subsequent algorithmic improvements 
%%yc delete
% by the authors
in~\cite{LiuEtAl2023} via a refined amortized analysis. The broader paradigm of explorable uncertainty has also inspired related oracle-based variants~\cite{DufosseDurrNadalTrystramVasquez2022} and applications in other online optimization domains, such as speed scaling~\cite{BampisDogeasKononovLucarelliPascual2021}.

Transitioning to multiple identical machines, optional-test scheduling has been studied extensively, primarily for makespan minimization. Albers and Eckl~\cite{AlbersEckl2021} explored optional-test scheduling on multiple identical machines in both non-preemptive and preemptive settings. In particular, they established lower bounds of $\max\left\{\varphi,\,2-\frac{1}{m}\right\}$ for the non-preemptive model and $\max\left\{\varphi,\,2-\frac{2}{m}+\frac{1}{m^2}\right\}$ for the preemptive model.
% both of which become asymptotically tight as $m\to\infty$. 
The non-preemptive upper bounds were subsequently improved 
%%yc delete
% by a series of works from the authors 
in~\cite{GongLin2022,GongGoebelLinMiyano2022,GongEtAl2025}. More closely related to our objective is the work 
%%yc delete
% of the authors
in~\cite{GongChenHayashi2024}, 
%%yc 
which studies
% who studied 
multiprocessor optional testing
%%yc 
for minimizing the total completion time. The authors obtained
% to minimize the total completion time. They achieved 
competitive ratios of $3.2361$ for arbitrary test times and $2.9271$ for uniform test times. 
%%yc add??
These bounds were later improved in~\cite{LiuEtAl2023}, where the corresponding ratios are $2.7763$ and $2.73606$, respectively.
However, compared to the makespan literature, the lower-bound side for multiprocessor testing under the sum-of-completion-times objective has remained largely underdeveloped.

% Distinct from optional testing, the obligatory-test model was only recently introduced by Dogeas et al.~\cite{DogeasErlebachLiang2024}. They focused on the single-machine case for minimizing the sum of completion times, proving upper bounds of $1.585$ (uniform) and $1.861$ (arbitrary), alongside a deterministic lower bound of $\sqrt{2}$. Prior to our work, this single-machine $\sqrt{2}$ bound was the only explicit deterministic lower bound for this objective in the obligatory model, with no comparable bounds known for multiple machines.
%\textcolor{blue}{
% While recent work by Buld and Schulz~\cite{BuldSchulz2025} has extended the optional testing model to multiple machines (for the weighted objective), the multiple-machine landscape for the obligatory model remains unexplored. Motivated by this, we focus on closing the upper and lower bound gaps for the unweighted total completion time in the multiprocessor obligatory-test setting.
%}
Distinct from optional testing, the obligatory-test model was introduced
%%yc delete
% by the authors 
in~\cite{DogeasErlebachLiang2024}. 
%%yc
The authors
% They
studied the single-machine case for minimizing the sum of completion times, proving upper bounds of $1.585$ for uniform test times and $1.861$ for arbitrary test times, together with a deterministic lower bound of $\sqrt{2}$.
Although the case \(m=1\) is formally contained in the parallel-machine model, this does not give a lower bound for any fixed \(m>1\). A naive replication of a single-machine hard instance does not force an online algorithm to keep the copies separated: the algorithm may interleave jobs from different copies and may schedule the test and processing operation of the same copied job on different machines. 
%%yc delete
% Thus, before our work, no direct lower-bound framework was known for the identical parallel-machine setting with fixed \(m>1\).
% Before our work, no direct lower-bound framework was known for the identical parallel-machine setting ($m>1$).
%%
The main difficulty is that single-machine exchange arguments do not provide a canonical way to compare tests and executions across several machine timelines. This motivates our completion-threshold approach, which reasons directly about global progress under total machine capacity.
% Most recently, Buld and Schulz~\cite{BuldSchulz2025} studied 
% %%kc0419
% %adversarial 
% optional
% %obligatory
% %%kc0419 change 
% testing for the total weighted completion time on both single and multiple machines.
% %%kc0419
% %Their work is complementary to ours, as we 
% We
% %%kc0419
% focus on closing the upper and lower bound gaps for the unweighted total completion time in the multiprocessor obligatory-test setting.

% \subparagraph{Paper Organization.}

% The remainder of the paper is organized as follows.
% Section~\ref{sec:pre} defines the model and introduces the notation used throughout our analysis. Section~\ref{sec:LB} develops the lower-bound framework, applying it to both three-type and multi-type adversarial constructions.  Section~\ref{sec:alg} presents our baseline deterministic algorithm, establishes its competitive guarantee, and discusses the limitations of the underlying priority family. Finally, we conclude in Section~\ref{sec:con} by discussing the remaining gap between the upper and lower bounds and outlining several open problems.

\subparagraph{Paper Organization.}
The remainder of the paper is organized as follows. Section~\ref{sec:pre} defines the model and introduces the notation used throughout the paper. Section~\ref{sec:overview} gives a technical overview of the lower-bound and upper-bound arguments. Section~\ref{sec:LB} develops the completion-threshold lower-bound framework and applies it to the three-type and dyadic multi-type constructions. Section~\ref{sec:alg} presents parallel 1-SORT and proves a lifting theorem from the single-machine guarantee to $m$ identical machines. Section~\ref{sec:con} concludes with open problems.

\section{Preliminaries}\label{sec:pre}

In this section, we formally define the obligatory-test scheduling on identical parallel machines and establish the 
%%yc
notation used throughout the paper.
% necessary notation. 
We then briefly review the 
%%yc
notion
% concept
of competitive analysis. Finally, we introduce the completion-threshold quantities $T_X$, which provide a mathematically convenient framework for our lower-bound arguments in Section~\ref{sec:LB}.

\subparagraph{Problem Definition.}
We study 
%%yc delete
% the problem of 
scheduling with obligatory testing on $m$ identical parallel machines, where a set of jobs $J = \{ 1,...,N\}$ is available at time $0$. Each job $j \in J$ requires two operations:
\begin{itemize}
    \item a test operation ($\tau_j$) of known length $t_j \ge 0$, and
    \item a processing operation ($\pi_j$) of unknown length $p_j \ge 0$.
\end{itemize}
The exact value of $p_j$ is revealed only upon completion of job $j$'s test.
In the obligatory-test model, testing is a strict prerequisite for processing. Thus, the processing operation of job $j$ cannot begin until its corresponding test has finished. The processing operation may be executed on any machine and does not need to start immediately after the test completes.
%%yc
In particular,
% (i.e., 
idle time or the execution of other operations may occur between a job's test and its processing operation.

We denote the total size of job $j$ by $\sigma_j := t_j + p_j$. 
The completion time of job $j$, denoted by $C_j$, is
%%yc delete
% defined as
the time at which its processing operation finishes. Our objective is to minimize the sum of completion times, $\sum_{j\in J} C_j$. We consider two standard test-time settings.
%%yc add
In
the uniform setting,
%%yc delete
% where
all jobs share the same test time,
%%yc
which we assume 
% assumed
without loss of generality to be $t_j = 1$ for all $j \in J$.
%%yc 
In
% and
the non-uniform setting, 
%%yc delete
% where
test times may vary across jobs. 
% For comparison, in the more widely studied optional-test model, the scheduler may choose to either test a job first or execute it directly using a known upper bound $u_j \ge p_j$ on its processing time.

\subparagraph{Competitive Ratio.}
Let $\ALG(I)$ denote the objective value achieved by 
%%yc
an
% the 
online algorithm on a given instance $I$, and let $\OPT(I)$ denote the value of an optimal offline schedule for the same instance. When the instance $I$ is clear from context, we abbreviate $\ALG(I)$ and $\OPT(I)$ as $\ALG$ and $\OPT$, respectively. 
%%yc delete
% We say that
A deterministic online algorithm is $\rho$-competitive if $\ALG(I)\le \rho\cdot \OPT(I)$ for every instance $I$.

\subparagraph{Completion Times and the Quantities \texorpdfstring{{\boldmath $T_X$}}{TX}.}

To facilitate our lower-bound arguments, 
%%yc
we
% it is useful to 
track the rate at which an algorithm completes jobs. Given a schedule produced by an online algorithm $A$, let $M_A(t)$ denote the number of jobs completed by time $t$. For each integer $X\in\{1,2,\dots,N\}$, 
%%yc delete
% we
define the completion threshold
%%yc delete
% $T_X$ as
\[
T_X := \min\{\, t \mid M_A(t)\ge X \,\}.
\]
In words, $T_X$ is the earliest time at which the schedule has completed at least $X$ jobs. The following lemma 
%%yc
shows
% demonstrates
that the total completion time can be exactly recovered by summing these threshold values.

\begin{lemma}\label{lem:tx-identity}
For every scheduling algorithm, $\ALG=\sum_{X=1}^N T_X$.
\end{lemma}
\begin{proof}
Let
\[
\theta_1<\theta_2<\cdots<\theta_q
\]
be the distinct completion times in the schedule, and let $a_r$ be the number of jobs that complete exactly at time $\theta_r$. Then
\[
\sum_{r=1}^q a_r = N
\qquad\text{and}\qquad
\ALG = \sum_{r=1}^q a_r \theta_r.
\]
Define
\[
A_r := \sum_{h=1}^r a_h
\qquad\text{and}\qquad
A_0:=0.
\]
By definition, exactly $A_r$ jobs have completed by time $\theta_r$. Therefore, for every
\[
X\in\{A_{r-1}+1,\dots,A_r\},
\]
the earliest time by which at least $X$ jobs have completed is exactly $\theta_r$. That is,
\[
T_X=\theta_r
\qquad
\text{for all } X\in\{A_{r-1}+1,\dots,A_r\}.
\]

Hence,
\[
\sum_{X=1}^N T_X
=
\sum_{r=1}^q \sum_{X=A_{r-1}+1}^{A_r} T_X
=
\sum_{r=1}^q \sum_{X=A_{r-1}+1}^{A_r} \theta_r
=
\sum_{r=1}^q a_r \theta_r
=
\ALG.
\]
This proves the claim.
\end{proof}
Lemma~\ref{lem:tx-identity} forms the foundation of our lower-bound framework in Section~\ref{sec:LB}. Since $\ALG=\sum_{X=1}^N T_X$, any lower bounds 
%%yc delete
% established
on the individual thresholds $T_X$ immediately translate 
%%yc
into
% to
a lower bound on the total completion time.

% \section{technical overview}
\section{Overview of the Techniques}\label{sec:overview}
We give a proof roadmap for the two main technical parts. 
%%yc
For the lower bound, we use completion thresholds to account for the amount of work that any online algorithm must have performed by the time it completes a given number of jobs. 
For the upper bound, we show that the parallel greedy schedule can be compared with the list schedule induced by the operation order of single-machine 1-SORT.
% The lower-bound proof is a work-accounting argument over completion thresholds. The upper-bound proof is an order-compatibility argument showing that the parallel greedy schedule can be analyzed as a list schedule induced by the single-machine 1-SORT order.

\subparagraph{Lower Bound.
% : accounting by the smallest completed type.
}Figure~\ref{fig:exchange-obstruction} illustrates the 
%%yc
difficulty of applying
% obstruction to
a direct exchange argument on multiple machines. 
\begin{figure}[t]
\centering
\includegraphics[width=.7\linewidth]{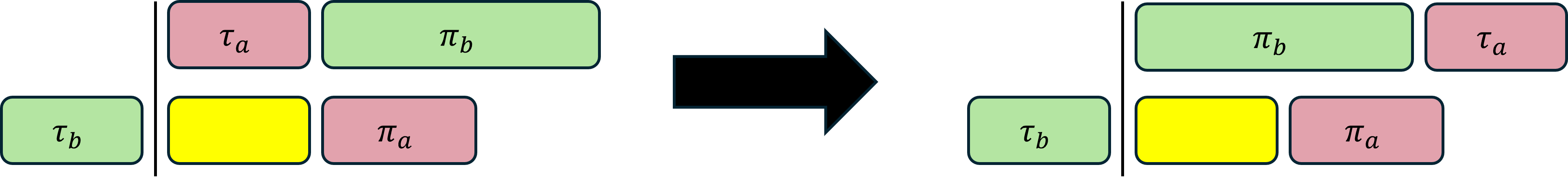}
\caption{
%%yc
A local difficulty for exchange arguments on multiple machines.
%%yc
Blocks of the same color belong to the same job, and the yellow block represents another operation in the schedule. 
% The yellow block denotes an arbitrary operation.
% A local exchange obstruction on multiple machines.  
Moving an execution operation from one machine timeline to another may change the execution order of other jobs. The lower-bound proof avoids this 
%%yc
difficulty
% issue
by using 
%%yc
completion thresholds
% threshold work accounting 
rather than exchange transformations.}
\label{fig:exchange-obstruction}
\end{figure}
Instead, we fix a completion threshold \(T_X\) and account for the amount of work that must have been completed by that time. In the \(K\)-type construction, a type-\(i\) job has unit test time and processing time \(i\), and \(\lambda_i\) denotes the fraction of jobs of type \(i\). The adversary reveals job types in decreasing order of processing time. Thus, if the smallest type among the first \(X\) completed jobs is \(r\), then all tests of types \(r+1,\ldots,K\) must have already been completed before any type-\(r\) job could 
%%yc delete
% even
be revealed. This gives the prefix cost
\[
    N\sum_{j=r+1}^{K}\lambda_j .
\]
%%yc
After this prefix has been accounted for,
% Moreover, after this prefix is paid, 
each of the \(X\) completed jobs still requires at least \(r+1\) units of work.  
This gives the baseline term \((r+1)X\). Finally, if \(X\) exceeds the total number of jobs of types \(r,\ldots,t-1\), then some completed jobs must have type at least \(t\).
Each such job adds one additional unit beyond the baseline. This yields the branch inequality
\[
    mT_X
    \ge
    N\sum_{j=r+1}^{K}\lambda_j
    +(r+1)X
    +
    \sum_{t=r+2}^{K}
    \left(
        X-N\sum_{j=r}^{t-1}\lambda_j
    \right)^+ .
\]
Normalizing \(x=X/N\), this becomes \(T_X\ge (N/m)f_r(x)\).  
Since the smallest completed type may be any \(r\), the usable threshold lower bound is the lower envelope
\[
    T_X
    \ge
    \frac{N}{m}g_K(X/N),
    \qquad
    g_K(x):=\min_{0\le r\le K}f_r(x).
\]

The dyadic choice of type proportions is designed so that this lower envelope can be computed explicitly. For
\[
    % \lambda_0=\frac12,\qquad
    \lambda_i=\frac{1}{2^{i+1}}\quad(0\le i\le K-1),
    \qquad
    \lambda_K=\frac{1}{2^K},
\]
write \(b_r(x)\) for the affine part of the branch \(f_r(x)\),
%%yc add
namely
\[
    b_r(x):=\sum_{j=r+1}^{K}\lambda_j+(r+1)x.
\]
Thus
\[
    b_r(x)=2^{-(r+1)}+(r+1)x
    \qquad (0\le r\le K-1),
\]
while
\[
    b_K(x)=(K+1)x.
\]
For \(0\le r\le K-2\), adjacent affine parts satisfy
\[
    b_{r+1}(x)-b_r(x)=x-2^{-(r+2)}.
\]
The final transition is
\[
    b_K(x)-b_{K-1}(x)=x-2^{-K}.
\]
These identities determine the active branch of the lower envelope. The branch \(f_K\) is active on \([0,2^{-K}]\).
For \(1\le r\le K-2\),
the branch \(f_r\) is active on
\[
    \bigl[2^{-(r+2)},\,2^{-(r+1)}\bigr];
\]
and \(f_0\) is active on \([1/4,1]\). Therefore
\[
    \ALG
    =
    \sum_{X=1}^{N}T_X
    \ge
    \frac{N}{m}\sum_{X=1}^{N}g_K(X/N)
    \ge
    \frac{N^2}{m}\int_0^1 g_K(x)\,dx.
\]
A direct evaluation gives
\[
    \int_0^1 g_K(x)\,dx
    =
    1-\frac{1}{2\cdot 4^K}.
\]
%%yc delete
% and hence, 
Since \(N=mn\),
%%yc add
this implies
\[
    \ALG
    \ge
    m\left(1-\frac{1}{2\cdot 4^K}\right)n^2 .
\]

The offline comparison uses the opposite structure.
Jobs are ordered by increasing total size.  
For the dyadic instance, the balanced offline schedule satisfies the recurrence
\[
    B_1=\frac58,
    \qquad
    B_K=\frac12+\frac14 B_{K-1}.
\]
%%yc
Hence,
\[
    B_K=\frac23-\frac{1}{6\cdot 4^K}.
\]
Thus
\[
    \OPT
    \le
    m\left(\frac23-\frac{1}{6\cdot 4^K}\right)n^2+O(n).
\]
Combining the two estimates gives
\[
    \liminf_{n\to\infty}\frac{\ALG}{\OPT}
    \ge
    \frac{1-\frac{1}{2\cdot 4^K}}
         {\frac23-\frac{1}{6\cdot 4^K}}
    =
    \frac{3(2-4^{-K})}{4-4^{-K}},
\]
which tends to \(3/2\).

\subparagraph{Upper Bound.
% : from parallel 1-SORT to an induced list schedule.
}The upper-bound algorithm is the parallel version of single-machine 1-SORT. At each event time, the algorithm first updates the available set.
Completed tests reveal the corresponding processing times, and the newly revealed execution operations become available. If \(r\) machines are idle, the algorithm starts up to \(r\) available unscheduled operations with smallest running times, where
\[
    \ell(\tau_j)=t_j,
    \qquad
    \ell(\pi_j)=p_j
\]
once \(\pi_j\) has been revealed. Ties are resolved by a fixed deterministic rule. Thus, parallel 1-SORT is a batched greedy rule.
It is 
%%yc delete
% locally
defined by the currently available operations, but it may start several operations simultaneously.
The 
%%yc add
main
technical issue is that the known guarantee for 1-SORT is stated through a single-machine operation order, while the parallel algorithm is defined through event batches. Let
\[
    L=(o_1,o_2,\ldots,o_{2N})
\]
be the order produced by single-machine 1-SORT on the same instance. 
We compare parallel 1-SORT with the \(m\)-machine list schedule induced by \(L\). At each event time, after updating the available set, the induced schedule scans \(L\) from left to right and starts the first available unscheduled operations on the idle machines.
The upper-bound proof has two 
%%yc
steps.
% layers.
First, we show that parallel 1-SORT is batch-equivalent to this induced list schedule. 
Second, after this reduction, we apply prefix-volume bounds to the induced list schedule. Let \(B_P(t)\) be the event batch chosen by parallel 1-SORT at time \(t\), and let \(B_L(t)\) be the event batch chosen by the list schedule induced by \(L\).  
Suppose, for contradiction, that \(t\) is the first event time at which the two batches differ. Choose
\[
    h\in B_P(t)\setminus B_L(t).
\]
Since the list schedule skips \(h\), there are at least \(r\) available operations before \(h\) in \(L\), where \(r\) is the number of idle machines.  
Since parallel 1-SORT chooses \(h\), among those earlier operations there is some \(o\) with
\[
    o\prec_L h
    \qquad\text{and}\qquad
    h\prec_P o.
\]

The proof then 
%%yc
derives
% forces
a contradiction 
%%yc
from
% through
two structural observations. 
First, \(h\) cannot be a test.
If it were a test, then it would have been available when the single-machine sequence chose \(o\), and the relation \(h\prec_P o\) would contradict \(o\prec_L h\). Hence \(h=\pi_j\) for some job \(j\).
Second, one must have
\[
    o\prec_L \tau_j\prec_L \pi_j=h.
\]
If instead \(\tau_j\prec_L o\), then \(\pi_j\) would already be available when the single-machine sequence chose \(o\), again contradicting \(h\prec_P o\). 
The remaining step is where the multi-machine nature matters. The prefix-clean invariant implies that no operation before \(\tau_j\) can remain as old available backlog after \(\tau_j\) has been scheduled. Together with monotonicity of test completion times along the order \(L\), this shows that \(h\), and every relevant operation before \(h\), must become newly available exactly at time \(t\). These newly available executions correspond to distinct tests completing at time \(t\), and each such test completion frees one machine. Thus, more machines are idle at time \(t\) than the assumed number \(r\), contradicting the fact that the list schedule skipped \(h\). This proves batch equivalence.

\begin{figure}[t]
\centering
\includegraphics[width=.4\linewidth]{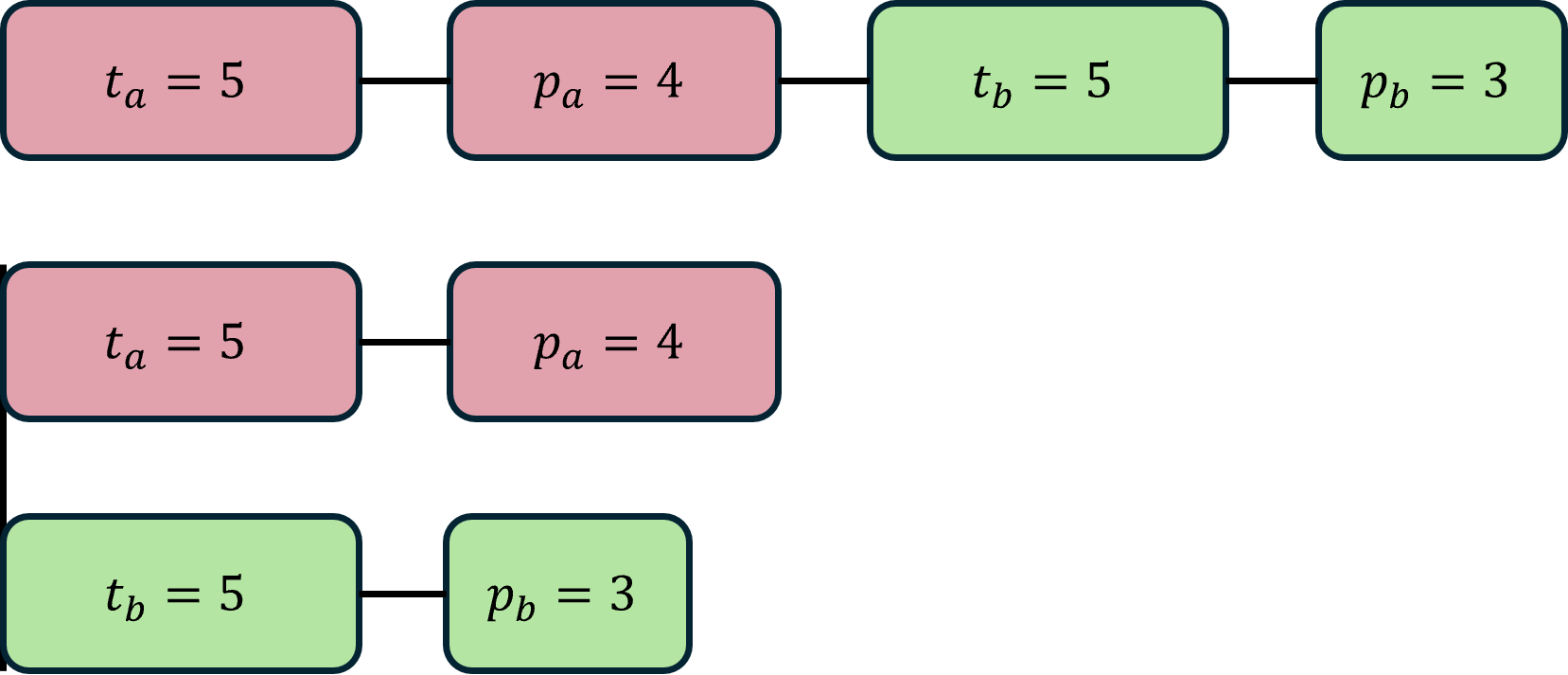}
\caption{A small instance illustrating batch equivalence. 
%%yc add
Blocks of the same color belong to the same job.
%%yc
The top row shows the linear operation order produced by single-machine 1-SORT, while the bottom rows show the corresponding event batches on two machines.
% The single-machine 1-SORT rule produces a linear operation order, whereas the parallel rule chooses simultaneous event batches. 
The batch-equivalence theorem proves that the \(m\)-machine list schedule induced by the single-machine order starts the same event batches as parallel 1-SORT.}
\label{fig:batch-equivalence-overview}
\end{figure}

Once batch equivalence is established, the completion-time analysis is short. For each job \(j\), define
\[
    U_j:=\sum_{q\prec_L\tau_j}\ell(q),
    \qquad
    V_j:=\sum_{\tau_j\prec_L q\prec_L\pi_j}\ell(q).
\]
In the single-machine 1-SORT schedule,
\[
    C_j^{(1)}=U_j+V_j+\sigma_j.
\]
In the induced \(m\)-machine list schedule, the test \(\tau_j\) can be delayed only by operations before it in \(L\), and the execution \(\pi_j\) is bounded by the usual prefix-volume argument.  This yields
\[
    C_j^{(m)}
    \le
    \frac{U_j+V_j}{m}+\sigma_j
    =
    \frac1m C_j^{(1)}
    +
    \left(1-\frac1m\right)\sigma_j.
\]
Summing over all jobs gives
\[
    \ALG_m
    \le
    \frac1m\ALG_1+
    \left(1-\frac1m\right)\sum_j\sigma_j.
\]
If single-machine 1-SORT is \(\rho\)-competitive, then \(\ALG_1\le \rho\OPT_1\).  Writing
\[
    A:=\sum_{j=1}^{N}j\sigma_j,
    \qquad
    B:=\sum_{j=1}^{N}\sigma_j,
\]
we have \(\OPT_1=A\), while the standard lower bound for \(m\) identical machines gives
\[
    \OPT_m
    \ge
    \frac1m A+\frac12\left(1-\frac1m\right)B.
\]
Therefore
\[
    \frac{\ALG_m}{\OPT_m}
    \le
    \frac{\rho A+(m-1)B}
         {A+\frac{m-1}{2}B}.
\]
Setting \(x=A/B\ge 1\), the right-hand side becomes
\[
    R(x)=\frac{\rho x+m-1}{x+\frac{m-1}{2}}.
\]
For \(\rho\le 2\), 
this function is nonincreasing on \([1,\infty)\), so
\[
    R(x)\le R(1)=\frac{2(m+\rho-1)}{m+1}.
\]
This proves the lifting theorem.

\section{Lower Bounds}\label{sec:LB}
This section proves the two lower-bound results stated in the introduction. We first give a three-type construction, which yields the numerical bound \(1.4811\).  We then give the dyadic construction, whose lower bound tends to \(3/2\).  Throughout this section, we write \((x)^+ := \max\{x,0\}\).

\subsection{A Three-type Lower Bound}
We begin with a three-type adversarial construction, which illustrates the
%%yc delete
% full
logic of the $T_X$ framework. All jobs have
%%yc delete
% a
unit test time, and their processing times are revealed according to the order in which they are tested. The construction separates jobs into long, medium, and short types.
%%yc 
It shows
% demonstrating 
how adversarial control of the revelation order translates into pointwise lower bounds on $T_X$, and consequently into a lower bound on $\sum_j C_j$.

\begin{theorem}\label{thm:three-type-lb}
Every deterministic online algorithm for obligatory-test scheduling on $m$ identical machines has an asymptotic competitive ratio of at least $\rho :=
\max_{\substack{0\le \alpha\le \beta\\ \alpha+\beta\le 1}}
\frac{\frac12+\alpha+\beta-\frac{\beta^2}{2}}
     {\frac12+\alpha^2+\alpha\beta+\frac{\beta^2}{2}}$. Numerically, this yields
     $\rho \ge 1.4811$, 
     %%yc
     attained at approximately 
     % where
     $\alpha \approx 0.1939, \beta \approx 0.2873$.
\end{theorem}

\begin{proof}
Fix parameters $\alpha,\beta$ satisfying $0\le \alpha\le \beta$ and $\alpha+\beta\le 1$. Let \(N=nm\).
%%yc ??
For notational simplicity, assume that the relevant quantities \(\alpha n\), \(\beta n\), and \((1-\alpha-\beta)n\) are integers.
% where \(n\) is chosen so that \(\alpha n\) and \(\beta n\) are integers. 
%%yc add??
The general case follows by rounding, which changes the objective values only by lower-order terms.
We construct an instance of $N$ jobs, all with test time $t_j=1$. 
%%kc
The adversary assigns processing times according to the order in which tests complete, revealing larger processing times earlier.
% To exploit the algorithm's lack of future information, the adversary forces the algorithm to accumulate completion times on the largest jobs first. 
%%yc
More precisely,
% Specifically, 
processing times are assigned according to the order in which tests complete,
%%yc
with simultaneous ties broken arbitrarily, so that shorter jobs are progressively revealed over time.
% (breaking simultaneous ties arbitrarily), progressively revealing shorter jobs as time goes on:
\begin{itemize}
    \item The first $\alpha N$ revealed jobs are of type $L$ (long) and have processing time $2$.
    \item The next $\beta N$ revealed jobs are of type $M$ (medium) and have processing time $1$.
    \item The remaining $(1-\alpha-\beta)N$ revealed jobs are of type $S$ (short) and have processing time $0$.
\end{itemize}

% For each $X\in\{1,2,\dots,N\}$, let $T_X$ denote the earliest time by which at least
% $X$ jobs have completed.
By Lemma~\ref{lem:tx-identity}, we have $\ALG=\sum_{X=1}^{N} T_X$. We first derive pointwise lower bounds on $T_X$.

\begin{lemma}\label{lem:three-type-TX}
For every $X\in\{1,2,\dots,N\}$,
\[
T_X \ge \frac{1}{m}
\begin{cases}
3X, & 1\le X\le \alpha N\\[1mm]
\alpha N + 2X, & \alpha N < X\le \beta N\\[1mm]
(\alpha+\beta)N + X, & \beta N < X\le (1-\alpha)N\\[1mm]
(2\alpha+\beta-1)N + 2X, & (1-\alpha)N < X\le N
\end{cases}
\]
\end{lemma}
\begin{proof}
Fix $X$ and consider the schedule up to time $T_X$.

\textbf{Case 1.} No type-$M$ or type-$S$ job has completed by time $T_X$.
Then all completed jobs are of type $L$.
Each such job requires one unit of testing and two units of processing, hence 
%%yc
three units of work.
% total work \(3\). 
Therefore, \(mT_X \ge 3X\).

\textbf{Case 2.} Some type-$M$ job has completed by time $T_X$, but no type-$S$ job has completed by time $T_X$.
Before the first type-$M$ job can be revealed, the algorithm must already have tested the first \(\alpha N\) jobs, using \(\alpha N\) units of work. After this prefix, each completed job requires at least two more units of work. A completed type-$L$ job requires its processing part of size \(2\), while a completed type-$M$ job requires one unit of testing and one unit of processing. Therefore, \(mT_X \ge \alpha N + 2X\).

\textbf{Case 3.} Some type-$S$ job has completed by time $T_X$.
Before a type-$S$ job can be revealed, the algorithm must already have tested the first
$\alpha N+\beta N$ jobs.
Thus,
$mT_X \ge (\alpha+\beta)N + X$.
This is the baseline bound in this case.
%%yc delete
% However, 
Among all jobs, the total number of non-$L$ jobs is
\[
\beta N + (1-\alpha-\beta)N = (1-\alpha)N.
\]
%%kc
Hence, if $X>(1-\alpha)N$, then among the $X$ completed jobs, at least
$X-(1-\alpha)N$ must be of type $L$.  The baseline term $X$ accounts for
one post-prefix unit of work for each completed job.  A completed type-$L$
job requires two post-prefix units of work, because its test has already
been counted in the prefix and its processing time is $2$.  Therefore each
such type-$L$ job 
%%yc
adds
% contributes
one additional unit beyond the baseline, and
we obtain $mT_X \ge (\alpha+\beta)N + X + \bigl(X-(1-\alpha)N\bigr)$.
% Hence, if $X>(1-\alpha)N$, then among the $X$ completed jobs, at least $X-(1-\alpha)N$ must be of type $L$. For each type-$L$ job, the baseline term denoted by ``$+X$'' represents only a single unit beyond the initial prefix $(\alpha+\beta)N$, whereas the completion of a type-$L$ job requires an additional two units following that prefix. Thus, each such type-$L$ job adds one extra unit, and we obtain
% $mT_X \ge (\alpha+\beta)N + X + \bigl(X-(1-\alpha)N\bigr)$.
Equivalently,
\[
mT_X \ge (2\alpha+\beta-1)N + 2X.
\]
Combining the above bounds yields
\[
mT_X \ge \min\Bigl\{3X,\;\alpha N+2X,\;(\alpha+\beta)N+X+\bigl(X-(1-\alpha)N\bigr)^+\Bigr\}.
\]
It remains to identify which expression is minimal on each interval.

If $1\le X\le \alpha N$, then
\[
\alpha N+2X \ge 3X,
\qquad
(\alpha+\beta)N+X \ge 2\alpha N + X \ge 3X,
\]
where the middle inequality uses $\alpha\le \beta$.
Hence,
$mT_X \ge 3X$.

If $\alpha N < X \le \beta N$, then the first case is excluded because there are only
$\alpha N$ type-$L$ jobs.
Moreover,
$
(\alpha+\beta)N + X \ge \alpha N + 2X
$
because $X\le \beta N$.
Hence,
$
mT_X \ge \alpha N + 2X.
$

If $\beta N < X \le (1-\alpha)N$, then the correction term is zero, and
$
\alpha N+2X \ge (\alpha+\beta)N+X
$
because $X\ge \beta N$.
Hence,
$mT_X \ge (\alpha+\beta)N + X$.

Finally, if $(1-\alpha)N < X\le N$, then the correction term is active, giving
\[
mT_X \ge (2\alpha+\beta-1)N + 2X.
\]
This proves the lemma.
\end{proof}

We now sum the bounds from Lemma~\ref{lem:three-type-TX}. Since $\ALG=\sum_{X=1}^{N}T_X$,
we obtain
\[
\ALG \ge \frac{1}{m}
\left(
\sum_{X=1}^{\alpha N} 3X
+
\sum_{X=\alpha N+1}^{\beta N} (\alpha N+2X)
+
\sum_{X=\beta N+1}^{(1-\alpha)N} ((\alpha+\beta)N+X)
+
\sum_{X=(1-\alpha)N+1}^{N} ((2\alpha+\beta-1)N+2X)
\right)
\]
A straightforward calculation yields $\ALG \ge m\left(\frac12+\alpha+\beta-\frac{\beta^2}{2}\right)n^2 + O(n)$.

For the offline optimum, we distribute the jobs evenly among the $m$ machines so that each machine receives exactly $\alpha n$ jobs of type $L$, $\beta n$ jobs of type $M$, and $(1-\alpha-\beta)n$ jobs of type $S$. Since the total sizes are respectively $3$, $2$, and $1$, the shortest-total-size order is $S \prec M \prec L$. 
On one machine, 
%%yc
the total completion time can be decomposed by job type as follows.
% the contributions of the three job types are
\begin{align*}
&S : \frac{(1-\alpha-\beta)^2}{2}n^2 + O(n)\\
&M :\beta(1-\alpha-\beta)n^2 + \beta^2 n^2 + O(n)\\
&L : \alpha(1-\alpha+\beta)n^2 + \frac32\alpha^2 n^2 + O(n)
\end{align*}
Hence, $\OPT \le m\left(\frac12+\alpha^2+\alpha\beta+\frac{\beta^2}{2}\right)n^2 + O(n)$. Therefore,
\[
\liminf_{n\to\infty}\frac{\ALG}{\OPT}
\ge
\frac{\frac12+\alpha+\beta-\frac{\beta^2}{2}}
     {\frac12+\alpha^2+\alpha\beta+\frac{\beta^2}{2}}.
\]
Maximizing over all feasible $(\alpha,\beta)$ proves the theorem.
\end{proof}

\subsection{A Multi-type Dyadic Lower Bound Approaching \texorpdfstring{{\boldmath $3/2$}}{3/2}}

We now present the main lower-bound construction. The guiding idea is to reveal processing times in many layers of decreasing magnitude.
%%yc
The dyadic proportions are
% using dyadic proportions 
chosen so that 
%%yc add
both
the lower-bound expression 
%%yc
for
% from
the online algorithm and the upper-bound structure of the offline optimum admit a clean recursive analysis. This multi-type construction is 
%%yc
where
% the point at which 
the $T_X$ framework becomes 
%%yc
particularly useful,
% genuinely stronger than the basic three-type case, 
leading to a deterministic lower bound approaching $3/2$.

\begin{definition}[The $K$-type dyadic adversarial family]
Fix an integer $K\ge 2$. Consider the obligatory-test scheduling problem on $m$ identical machines with $N=nm$ jobs, all having test time $1$. 
%%yc add??
We assume, for notational simplicity, that \(n\) is divisible by \(2^K\). The general case follows by rounding the number of jobs of each type, which affects all objective values only by lower-order terms.
The adversary defines $K+1$ job types.
%%yc delete
% where
A type-$i$ job,
%%yc add
where
$i\in\{0,1,\dots,K\}$,
has processing time $i$. 
%%yc
The family is defined as follows.
% This adversarial family is defined by the following rules:
\begin{itemize}
    \item \textbf{Type Proportions:} The proportion of type-$i$ jobs is 
    %%yc delete
    % given by
    \[
    % \lambda_0=\frac12,\qquad
    \lambda_i=\frac{1}{2^{i+1}}\quad \text{for } 0\le i\le K-1,\qquad
    \lambda_K=\frac{1}{2^K},
    \]
    so that $\sum_{i=0}^K \lambda_i=1$.
    \item \textbf{Assignment Rule:} The adversary assigns processing times dynamically according to the exact order in which tests complete,
    %%yc
    with simultaneous ties broken arbitrarily.
    % (breaking simultaneous ties arbitrarily).
    \item \textbf{Revelation Sequence:} The adversary reveals the job types in strictly decreasing order of processing time. The first $\lambda_K N$ tested jobs are revealed to be of type $K$, the next $\lambda_{K-1}N$ jobs are revealed to be of type $K-1$, and so on, until the final $\lambda_0N$ tested jobs are revealed to be of type $0$.
\end{itemize}
\end{definition}

\begin{theorem}\label{thm:multi-type-dyadic-lb}
For every integer $K\ge 2$, under the $K$-type dyadic adversarial family defined above, every deterministic online algorithm satisfies
\[
\liminf_{n\to\infty}\frac{\ALG}{\OPT}
\ge
\frac{3(2-4^{-K})}{4-4^{-K}}.
\]
In particular, as $K\to\infty$, this lower bound tends to $\frac32$.
\end{theorem}

The proof has four steps. First, for each completion threshold $T_X$, we lower-bound the amount of work required when the smallest completed type is $r$. Second, this yields a lower envelope of the form $T_X \ge \frac{N}{m} g_K(X/N)$. Third, we integrate this lower envelope to obtain an explicit lower bound on $\ALG$. Finally, we upper-bound $\OPT$ by a recursive balanced schedule whose cost satisfies a simple recurrence. 

To prove Theorem~\ref{thm:multi-type-dyadic-lb}, we analyze the completion thresholds $T_X$. 
%%yc
Define
% Defining
the normalized index $x:=X/N\in[0,1]$.
Lemma~\ref{lem:tx-identity} gives $\ALG=\sum_{X=1}^{N}T_X$. We first derive a lower bound on $T_X$ depending on the smallest type among the jobs completed by time $T_X$.

\begin{lemma}\label{lem:general-r-branch}
Fix $X\in\{1,\dots,N\}$, and suppose that by time $T_X$ the smallest type among the completed jobs is exactly $r$, where $0\le r\le K$.
Then $mT_X
\ge
N\sum_{j=r+1}^{K}\lambda_j
+
(r+1)X
+
\sum_{t=r+2}^{K}
\left(
X - N\sum_{j=r}^{t-1}\lambda_j
\right)^+$.
\end{lemma}

\begin{proof}
Fix $X$ and assume that among the jobs completed by time $T_X$, the minimum type is $r$.
%%yc delete
% First, 
Before a type-$r$ job can be revealed, the algorithm must already have tested all jobs of types $r+1,r+2,\dots,K$. Therefore, the schedule must spend at least $N\sum_{j=r+1}^{K}\lambda_j$ units of work on these tests before any completed type-$r$ job can exist.

Now consider the $X$ completed jobs themselves. Every such job has type at least $r$. After accounting for these prefix tests, each completed job still requires at least $r+1$ additional units of work.
\begin{itemize}
    \item A completed type-$r$ job still needs its own test and its processing part, for a total length of $r+1$.
    \item A completed type-$s$ job with $s>r$ has already had its test counted in the prefix,
    but still requires processing time $s\ge r+1$.
\end{itemize}
Hence, we obtain the baseline term $(r+1)X$. This baseline can be strengthened when the first \(X\) completed jobs cannot all come from types \(r,r+1,\dots,t-1\). Fix $t\in\{r+2,\dots,K\}$. The total number of jobs of types $r,r+1,\dots,t-1$ is $N\sum_{j=r}^{t-1}\lambda_j$.
%%yc delete
% Therefore, 
If $X > N\sum_{j=r}^{t-1}\lambda_j$, then among the first $X$ completed jobs, at least $X - N\sum_{j=r}^{t-1}\lambda_j$ must have type at least $t$. For each such job, the baseline term accounts for only $r+1$ units, whereas a completed type-$t$ or higher job requires at least one more unit.
%%yc delete
% beyond this.
Thus, we obtain the correction term $\left( X - N\sum_{j=r}^{t-1}\lambda_j \right)^+$. Summing the baseline and all correction terms proves the claim.
\end{proof}

\begin{definition}
For each \(r\in\{0,1,\dots,K\}\), define
$
b_r(x):=(r+1)x+\sum_{j=r+1}^{K}\lambda_j
$ and
$
f_r(x)
:=
b_r(x)
+
\sum_{t=r+2}^{K}
\left(
x-\sum_{j=r}^{t-1}\lambda_j
\right)^+ .
$
When \(r\ge K-1\), this
%%yc add
final summation is interpreted as zero. Finally, define
$
g_K(x):=\min_{0\le r\le K} f_r(x).
$
\end{definition}

\begin{corollary}\label{cor:TX-envelope}
For every $X\in\{1,\dots,N\}$, $T_X \ge \frac{N}{m}\, g_K\!\left(\frac{X}{N}\right)$.
\end{corollary}

% We now specialize the branches $f_r$ to the dyadic choice of the proportions. Recall that 
% $\lambda_0=\frac12$,
% $\lambda_i=\frac1{2^{i+1}}$ for all $1\le i\le K-1$, and
% $\lambda_K=\frac1{2^K}$. Then $f_r(x) = b_r(x)+\sum_{t=r+2}^{K}\Bigl(x-\sum_{j=r}^{t-1}\lambda_j\Bigr)^+$
% for
% $\ 0\le r\le K$,  with the convention that the summation over \(t\) is zero when $r=K-1$ or $r=K$.
% In the dyadic case, this becomes $b_r(x)=2^{-(r+1)}+(r+1)x$
% for
% $0\le r\le K-1$,
% and 
% $b_K(x)=(K+1)x$. 
% Moreover, for every \(0\le q\le K-2\), $b_{q+1}(x)-b_q(x)=x-2^{-(q+2)}$. The final transition is slightly different: $b_K(x)-b_{K-1}(x)=x-2^{-K}$.

We now specialize these quantities to the dyadic choice of proportions. Recall that
$
    \lambda_i=\frac1{2^{i+1}}\quad(0\le i\le K-1),
    \lambda_K=\frac1{2^K}.
$
Then
$
    b_r(x)=2^{-(r+1)}+(r+1)x,\  (0\le r\le K-1),
$
while
$
    b_K(x)=(K+1)x.
$
Moreover, for every \(0\le q\le K-2\),
$
    b_{q+1}(x)-b_q(x)=x-2^{-(q+2)},
$
and the final transition is
$
    b_K(x)-b_{K-1}(x)=x-2^{-K}.
$
\begin{lemma}\label{lem:dyadic-left-envelope}
For the dyadic family above, the following hold.
\begin{enumerate}
    \item On the interval $0\le x\le 2^{-K}$, we have $g_K(x)=f_K(x)=(K+1)x$.
    \item For every integer $r$ with $1\le r\le K-2$, on the interval $I_r:=\bigl[2^{-(r+2)},\,2^{-(r+1)}\bigr]$, we have $g_K(x)=f_r(x)=2^{-(r+1)}+(r+1)x$.
\end{enumerate}
\end{lemma}

\begin{proof}
We first note that for every \(0\le s\le K-1\),
\[
    \sum_{j=s+1}^{K}\lambda_j = 2^{-(s+1)}.
\]
Hence, whenever no correction term is active in \(f_s\), we have
\[
    f_s(x)=b_s(x)=2^{-(s+1)}+(s+1)x.
\]
We start with part~(2). Fix \(r\in\{1,\dots,K-2\}\) and let
\[
    I_r=\bigl[2^{-(r+2)},\,2^{-(r+1)}\bigr].
\]
%%yc
For every \(t\) with \(r+2\le t\le K\),
% For every \(K\ge t\ge r+2\),
\[
\sum_{j=r}^{t-1}\lambda_j
=
\sum_{j=r}^{t-1}2^{-(j+1)}
=
2^{-r}-2^{-t}
\ge
2^{-r}-2^{-(r+2)}
=
3\cdot 2^{-(r+2)}
>
2^{-(r+1)}.
\]
Therefore, for every \(x\in I_r\) and every 
%%yc
\(t\) with \(r+2\le t\le K\),
% \(K\ge t\ge r+2\),
\[
    x\le 2^{-(r+1)}<\sum_{j=r}^{t-1}\lambda_j.
\]
%%yc 
Thus, every correction term in \(f_r\) is zero. Hence,
\[
    f_r(x)=b_r(x)=2^{-(r+1)}+(r+1)x.
\]

We now show that \(f_r(x)\le f_s(x)\) for every \(s\neq r\).
If \(s<r\), then all correction terms in \(f_s\) are nonnegative, so
\(f_s(x)\ge b_s(x)\). Also,
\[
f_r(x)-b_s(x)
=
b_r(x)-b_s(x)
=
\sum_{q=s}^{r-1}\bigl(b_{q+1}(x)-b_q(x)\bigr)
=
\sum_{q=s}^{r-1}\bigl(x-2^{-(q+2)}\bigr).
\]
Since \(q\le r-1\), we have \(2^{-(q+2)}\ge 2^{-(r+1)}\ge x\) for every
\(x\in I_r\). Hence every summand is nonpositive, and therefore
\[
    f_r(x)\le b_s(x)\le f_s(x).
\]

If \(s>r\), then again all correction terms in \(f_s\) are nonnegative, so
\(f_s(x)\ge b_s(x)\). We show that \(b_s(x)\ge b_r(x)\).
First suppose that \(s\le K-1\). Then
\[
b_s(x)-b_r(x)
=
\sum_{q=r}^{s-1}\bigl(b_{q+1}(x)-b_q(x)\bigr)
=
\sum_{q=r}^{s-1}\bigl(x-2^{-(q+2)}\bigr).
\]
Since \(q\ge r\), we have \(2^{-(q+2)}\le 2^{-(r+2)}\le x\) for every
\(x\in I_r\). Hence every summand is nonnegative, and so
\(b_s(x)\ge b_r(x)\).

It remains to consider \(s=K\). In this case,
\[
b_K(x)-b_r(x)
=
\sum_{q=r}^{K-2}\bigl(x-2^{-(q+2)}\bigr)
+
\bigl(x-2^{-K}\bigr).
\]
The terms in the first sum are nonnegative by the same argument as above.
Moreover, since \(r\le K-2\) and \(x\in I_r\), we have
\[
    x\ge 2^{-(r+2)}\ge 2^{-K}.
\]
Thus the final term \(x-2^{-K}\) is also nonnegative. Hence
\(b_K(x)\ge b_r(x)\).

In both cases \(b_s(x)\ge b_r(x)=f_r(x)\), and therefore
\[
    f_r(x)\le b_s(x)\le f_s(x).
\]
Thus, \(f_r\) is the minimum branch on \(I_r\), which proves part~(2).

For part~(1), fix \(x\in[0,2^{-K}]\). Since
\[
    f_K(x)=b_K(x)=(K+1)x,
\]
it suffices to show that \(f_K(x)\le f_s(x)\) for every \(0\le s\le K-1\).
For each such \(s\), all correction terms in \(f_s\) are nonnegative, so
\(f_s(x)\ge b_s(x)\). Also,
\[
f_K(x)-b_s(x)
=
b_K(x)-b_s(x)
=
\sum_{q=s}^{K-2}\bigl(x-2^{-(q+2)}\bigr)
+
\bigl(x-2^{-K}\bigr).
\]
For every \(q\le K-2\), we have
\[
    2^{-(q+2)}\ge 2^{-K}\ge x.
\]
Also \(x-2^{-K}\le 0\). Hence every term in the displayed expression is
nonpositive, and therefore
\[
    f_K(x)\le b_s(x)\le f_s(x).
\]
Thus, \(f_K\) is the minimum branch on \([0,2^{-K}]\), proving part~(1).
\end{proof}

\begin{lemma}\label{lem:f0-dominates}
On the interval $\frac14 \le x \le 1$, we have $g_K(x)=f_0(x)$, where $f_0(x)=\frac12+x+\sum_{t=2}^{K}\bigl(x-(1-2^{-t})\bigr)^+$.
\end{lemma}

\begin{proof}
By definition of the dyadic proportions,
\[
\sum_{j=1}^{K}\lambda_j=\frac12,
\qquad
\sum_{j=0}^{t-1}\lambda_j = 1-2^{-t}
\quad (t\ge 2).
\]
Therefore the formula for $f_0$ is immediate. We must show that
\[
f_0(x)\le f_r(x)
\qquad
\text{for every }r\ge 1\text{ and every }x\in[1/4,1].
\]
For every $r\ge 1$, all correction terms in $f_r$ are nonnegative, so
$
f_r(x)\ge b_r(x).
$
We next show that, on \([1/4,1]\), the affine parts \(b_r\) are minimized at
\(r=1\) among all \(r\ge1\).  First let \(2\le r\le K-1\). Then
\[
b_r(x)-b_1(x)
=
\sum_{q=1}^{r-1}\bigl(b_{q+1}(x)-b_q(x)\bigr)
=
\sum_{q=1}^{r-1}\bigl(x-2^{-(q+2)}\bigr).
\]
If \(x\ge 1/4\), then for every \(q\ge1\),
\[
2^{-(q+2)}\le 2^{-3}<\frac14\le x,
\]
Thus every summand is nonnegative. Hence \(b_r(x)\ge b_1(x)\) for
\(2\le r\le K-1\).

It remains to compare \(b_K\) with \(b_1\). We have
\[
b_K(x)-b_1(x)
=
\sum_{q=1}^{K-2}\bigl(x-2^{-(q+2)}\bigr)
+
\bigl(x-2^{-K}\bigr).
\]
The first sum is nonnegative by the same argument, and
\(x-2^{-K}\ge0\) for \(x\in[1/4,1]\). Therefore
\[
b_r(x)\ge b_1(x)
\qquad
\text{for all }r\ge2\text{ and }x\in[1/4,1].
\]
First consider \(x\in[1/4,3/4)\).  In this range no correction term in \(f_0\) is active, and hence
\[
    f_0(x)=\frac12+x.
\]
Therefore
\[
    b_1(x)-f_0(x)
    =
    \left(\frac14+2x\right)-\left(\frac12+x\right)
    =
    x-\frac14
    \ge 0.
\]

Next suppose that \(x\in [3/4,1-2^{-K})\).  This case is nonempty only when \(K\ge3\).  Let \(q\) be the unique integer with
\(2\le q\le K-1\) such that
\[
    1-2^{-q}\le x<1-2^{-(q+1)}.
\]
Then exactly the terms with \(2\le t\le q\) are active in \(f_0\).  Hence
\[
    f_0(x)
    =
    \frac12+x+\sum_{t=2}^{q}\bigl(x-(1-2^{-t})\bigr)
    =
    qx-q+2-2^{-q}.
\]
Thus
\[
    b_1(x)-f_0(x)
    =
    \frac14+2x-\bigl(qx-q+2-2^{-q}\bigr).
\]
Since the coefficient of \(x\) is \(2-q\le0\), this difference is minimized at
\(x=1-2^{-(q+1)}\), where it equals
\[
    \frac14+\frac{q}{2^{q+1}}>0.
\]

It remains to consider \(x\in[1-2^{-K},1]\).  In this interval all terms
\(2\le t\le K\) are active, so
\[
    f_0(x)
    =
    \frac12+x+\sum_{t=2}^{K}\bigl(x-(1-2^{-t})\bigr)
    =
    Kx-K+2-2^{-K}.
\]
Therefore
\[
    b_1(x)-f_0(x)
    =
    \frac14+2x-\bigl(Kx-K+2-2^{-K}\bigr).
\]
The coefficient of \(x\) is \(2-K\le0\), so the minimum on
\([1-2^{-K},1]\) is attained at \(x=1\).  Hence
\[
    b_1(x)-f_0(x)
    \ge
    b_1(1)-f_0(1)
    =
    \frac14+2^{-K}>0.
\]
Consequently, \(f_0(x)\le b_1(x)\le b_r(x)\le f_r(x)\) for every \(r\ge1\), and therefore
\[
    g_K(x)=f_0(x)
    \qquad
    \text{for every }x\in[1/4,1].
\]
\end{proof}

\begin{figure}[t]
\centering
\begin{tikzpicture}
\begin{axis}[
    width=.96\linewidth,
    height=6.4cm,
    xmin=0, xmax=1,
    ymin=0, ymax=2.15,
    axis lines=left,
    xlabel={$x$},
    ylabel={$f_r(x),\,g_4(x)$},
    domain=0:1,
    samples=400,
    xtick={0,1/16,1/8,1/4,1/2,3/4,1},
    xticklabels={$0$,$\frac1{16}$,$\frac18$,$\frac14$,$\frac12$,$\frac34$,$1$},
    ytick={0,1/4,1/2,1,3/2,2},
    yticklabels={$0$,$\frac14$,$\frac12$,$1$,$\frac32$,$2$},
    legend style={
        font=\scriptsize,
        at={(0.03,0.97)},
        anchor=north west,
        draw=none,
        fill=white
    },
    clip=true,
]

% individual branches
\addplot[thick, blue] {1/2 + x + max(x-3/4,0) + max(x-7/8,0) + max(x-15/16,0)};
\addlegendentry{$f_0$}

\addplot[thick, red] {1/4 + 2*x + max(x-3/8,0) + max(x-7/16,0)};
\addlegendentry{$f_1$}

\addplot[thick, teal] {1/8 + 3*x + max(x-3/16,0)};
\addlegendentry{$f_2$}

\addplot[thick, violet] {1/16 + 4*x};
\addlegendentry{$f_3$}

\addplot[thick, orange] {5*x};
\addlegendentry{$f_4$}

% lower envelope g_4
\addplot[very thick, black] {
    (x <= 1/16) ? (5*x) : (
    (x <= 1/8)  ? (1/8 + 3*x) : (
    (x <= 1/4)  ? (1/4 + 2*x) : 
                  (1/2 + x + max(x-3/4,0) + max(x-7/8,0) + max(x-15/16,0))
    ))
};
\addlegendentry{$g_4=\min_{0\le r\le 4} f_r$}

% transition markers
\addplot[densely dashed, gray] coordinates {(1/16,0) (1/16,2.15)};
\addplot[densely dashed, gray] coordinates {(1/8,0) (1/8,2.15)};
\addplot[densely dashed, gray] coordinates {(1/4,0) (1/4,2.15)};

\end{axis}
\end{tikzpicture}
\caption{The dyadic branches $f_0,\dots,f_4$ and the lower envelope $g_4$ for $K=4$.
The 
%%yc
function \(g_4\)
% $g_K(x)$ 
follows $f_4$ on $[0,1/16]$, $f_2$ on $[1/16,1/8]$, $f_1$ on $[1/8,1/4]$,
and $f_0$ on $[1/4,1]$.}
\label{fig:dyadic-envelope-k4}
\end{figure}
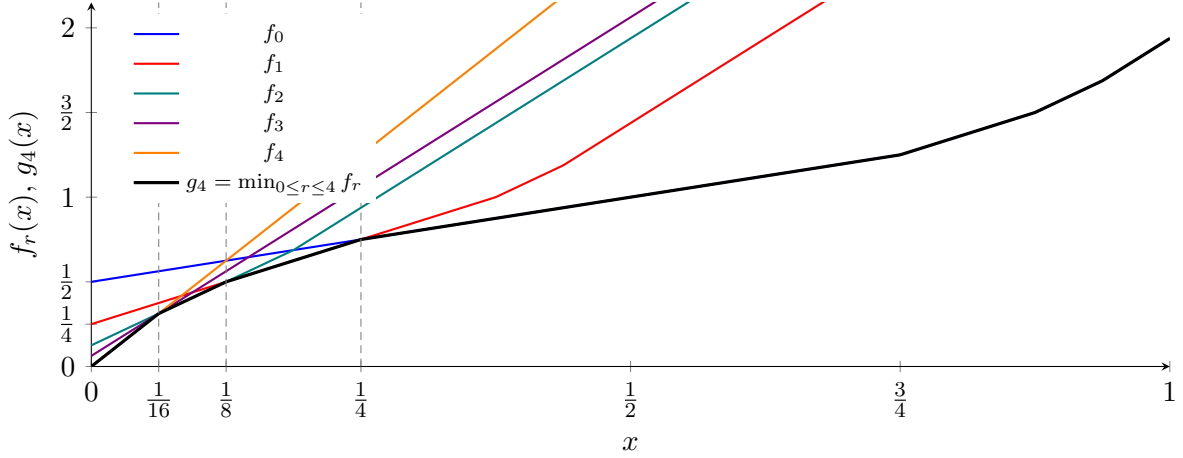

We now integrate $g_K(x)$. The previous two lemmas identify $g_K$ explicitly. We distinguish the corresponding intervals below.
\[
g_K(x)=
\begin{cases}
(K+1)x, & 0\le x\le 2^{-K},\\[1mm]
2^{-(r+1)}+(r+1)x,
& \text{for the unique } r\in\{1,\dots,K-2\}
  \text{ such that } 2^{-(r+2)}\le x\le 2^{-(r+1)},\\[1mm]
% 2^{-(r+1)}+(r+1)x,
% & 2^{-(r+2)}\le x\le 2^{-(r+1)},\quad 1\le r\le K-2,\\[1mm]
f_0(x), & \frac14\le x\le 1.
\end{cases}
\]
The key point 
%%yc add??
for passing from the discrete sum to the integral
is that $g_K$ is nondecreasing on $[0,1]$.

\begin{lemma}\label{lem:integral-lb}
Let $A_K:=\int_0^1 g_K(x)\,dx$. Then $\ALG \ge \frac{N^2}{m} A_K = mA_K n^2$.
\end{lemma}

\begin{proof}
By Corollary~\ref{cor:TX-envelope}, $\ALG = \sum_{X=1}^{N}T_X \ge \frac{N}{m}\sum_{X=1}^{N} g_K\!\left(\frac{X}{N}\right)$.
Since $g_K$ is nondecreasing, for every $X\in\{1,\dots,N\}$, we have $\int_{(X-1)/N}^{X/N} g_K(x)\,dx \le\frac{1}{N} g_K\!\left(\frac{X}{N}\right)$. Summing over all $X$ yields $\int_0^1 g_K(x)\,dx \le \frac1N \sum_{X=1}^{N} g_K\!\left(\frac{X}{N}\right)$. Multiplying this inequality by $N^2/m$ yields the stated bound.
\end{proof}

\begin{lemma}\label{lem:AK-value}
For every $K\ge 2$, $A_K = 1-\frac{1}{2\cdot 4^K}$.
\end{lemma}

\begin{proof}
We split the integral into the left dyadic part and the right $f_0$ part.
\subparagraph*{Left part.}
By Lemma~\ref{lem:dyadic-left-envelope},
\[
\int_0^{1/4} g_K(x)\,dx
=
\int_0^{2^{-K}} (K+1)x\,dx
+
\sum_{r=1}^{K-2}\int_{2^{-(r+2)}}^{2^{-(r+1)}} \bigl(2^{-(r+1)}+(r+1)x\bigr)\,dx.
\]
The first integral is
\[
\int_0^{2^{-K}} (K+1)x\,dx
=
\frac{K+1}{2^{2K+1}}.
\]
For $1\le r\le K-2$,
\[
\int_{2^{-(r+2)}}^{2^{-(r+1)}} \bigl(2^{-(r+1)}+(r+1)x\bigr)\,dx
=
\frac{3r+7}{2^{2r+5}}.
\]
Therefore,
\[
\int_0^{1/4} g_K(x)\,dx
=
\frac{K+1}{2^{2K+1}}
+
\sum_{r=1}^{K-2}\frac{3r+7}{2^{2r+5}}.
\]
A direct summation gives
\[
\int_0^{1/4} g_K(x)\,dx
=
\frac{11}{96}-\frac{1}{3\cdot 4^K}.
\]

\subparagraph*{Right part.}
By Lemma~\ref{lem:f0-dominates},
\[
\int_{1/4}^{1} g_K(x)\,dx
=
\int_{1/4}^{1}\left(\frac12+x\right)\,dx
+
\sum_{t=2}^{K}\int_{1-2^{-t}}^{1}\bigl(x-(1-2^{-t})\bigr)\,dx.
\]
The first term equals
\[
\int_{1/4}^{1}\left(\frac12+x\right)\,dx=\frac{27}{32}.
\]
For each $t\ge 2$,
\[
\int_{1-2^{-t}}^{1}\bigl(x-(1-2^{-t})\bigr)\,dx
=
\frac{1}{2^{2t+1}}.
\]
Hence,
\[
\sum_{t=2}^{K}\frac{1}{2^{2t+1}}
=
\frac12\sum_{t=2}^{K}4^{-t}
=
\frac{1}{24}\left(1-4^{-(K-1)}\right)
=
\frac{1}{24}-\frac{1}{6\cdot 4^K}.
\]
Thus,
\[
\int_{1/4}^{1} g_K(x)\,dx
=
\frac{27}{32}+\frac{1}{24}-\frac{1}{6\cdot 4^K}
=
\frac{85}{96}-\frac{1}{6\cdot 4^K}.
\]

Adding the two parts, we obtain
\[
A_K
=
\left(\frac{11}{96}-\frac{1}{3\cdot 4^K}\right)
+
\left(\frac{85}{96}-\frac{1}{6\cdot 4^K}\right)
=
1-\frac{1}{2\cdot 4^K}.
\]
This completes the proof.
\end{proof}

%%yc add??
We record the resulting online lower bound before turning to the offline comparison. Combining Lemma~\ref{lem:integral-lb} and Lemma~\ref{lem:AK-value} gives
$\ALG \ge m\left(1-\frac{1}{2\cdot 4^K}\right)n^2$. We now upper-bound the offline optimum by exhibiting a recursive schedule.

\begin{lemma}\label{lem:BK-recurrence}
For each $K\ge 1$, there exists an offline schedule with an objective value at most $m\bigl(B_K n^2 + O(n)\bigr)$, where $B_1=\frac58, B_K=\frac12+\frac14 B_{K-1}$ for $K\ge 2$. Consequently, $B_K=\frac23-\frac{1}{6\cdot 4^K}$.
\end{lemma}

\begin{proof}
Distribute the jobs evenly over the $m$ machines so that each machine receives exactly
\[
\lambda_0 n,\lambda_1 n,\dots,\lambda_K n
\]
jobs of types $0,1,\dots,K$, respectively.
%%yc 
On each machine, schedule the jobs
% and schedule them on each machine 
in nondecreasing order of total size,
$
0,1,2,\dots,K
$. It suffices to analyze one machine.

There are exactly $n/2$ type-$0$ jobs. Each has total size $1$, so their own cost is
\[
1+2+\cdots+\frac{n}{2}
=
\frac18 n^2 + O(n).
\]
These type-$0$ jobs also delay each of the remaining $n/2$ jobs by exactly $n/2$ time units, adding
\[
\frac{n}{2}\cdot \frac{n}{2}
=
\frac14 n^2.
\]

Now consider the remaining jobs of types $1,\dots,K$. If we subtract one unit from the total size of every remaining job, then their type proportions become exactly
\[
\frac12,\frac14,\dots,\frac{1}{2^{K-1}}.
\]
This is the $(K-1)$-layer dyadic family, scaled down from $n$ jobs to $n/2$ jobs on this machine. Hence, the residual cost is
\[  
 B_{K-1} (\frac n2)^2 + O(n) = \frac14 B_{K-1} n^2 + O(n).
\]
Finally, adding back the removed unit to each of the remaining $n/2$ jobs adds
\[
1+2+\cdots+\frac{n}{2}
=
\frac18 n^2 + O(n).
\]

Summing the four terms yields
\[
B_K
=
\frac18+\frac14+\frac18+\frac14 B_{K-1}
=
\frac12+\frac14 B_{K-1}.
\]

For $K=1$, one machine has $n/2$ jobs of size $1$ and $n/2$ jobs of size $2$. Thus,
\[
B_1
=
\frac18+\frac14+\frac14
=
\frac58.
\]

To solve the recurrence, subtract $2/3$ from both sides.
\[
B_K-\frac23
=
\frac14\left(B_{K-1}-\frac23\right)
\]
Since
\[
B_1-\frac23 = \frac58-\frac23 = -\frac{1}{24},
\]
it follows that
\[
B_K-\frac23
=
-\frac{1}{24\cdot 4^{K-1}}
=
-\frac{1}{6\cdot 4^K}.
\]
Hence,
\[
B_K=\frac23-\frac{1}{6\cdot 4^K}.
\]
\end{proof}

\begin{proof}[Proof of Theorem~\ref{thm:multi-type-dyadic-lb}]
Combining Lemma~\ref{lem:integral-lb} and Lemma~\ref{lem:AK-value} yields
$\ALG \ge m\left(1-\frac{1}{2\cdot 4^K}\right)n^2$. By Lemma~\ref{lem:BK-recurrence}, $\OPT \le m\left(\frac23-\frac{1}{6\cdot 4^K}\right)n^2 + O(n)$. Evaluating the asymptotic competitive ratio gives
$$
\liminf_{n\to\infty}\frac{\ALG}{\OPT} \ge \frac{1-\frac{1}{2\cdot 4^K}}{\frac23-\frac{1}{6\cdot 4^K}} = \frac{3(2-4^{-K})}{4-4^{-K}}.
$$
Taking the limit as $K\to\infty$ yields $\frac32$, which concludes the proof.
\end{proof}

\begin{corollary}
No deterministic online algorithm for obligatory-test scheduling on a single machine can achieve an asymptotic competitive ratio strictly less than $\frac32$.
\end{corollary}

\section{Algorithm}
\label{sec:alg}

We now present the algorithmic upper bound. The algorithm is the natural parallel version of the single-machine 1-SORT rule from~\cite{DogeasErlebachLiang2024}. On one machine, 1-SORT always selects an available operation with minimum running time.  
We extend this rule directly to $m$ identical machines. Whenever machines become idle, the algorithm assigns them available operations with smallest running times.
For a test operation $\tau_j$, the running time is the known test time $t_j$. For an execution operation $\pi_j$, the running time is the revealed processing time $p_j$, which is known only after $\tau_j$ has completed. Ties are broken by a fixed deterministic rule. We call the resulting algorithm \emph{parallel 1-SORT}.
The goal of this section is to prove a lifting theorem from one machine to $m$ machines. If single-machine 1-SORT is $\rho$-competitive, then parallel 1-SORT is $\frac{2(m+\rho-1)}{m+1}$ -competitive.  
For an operation $o$, we write $S(o)$ and $C(o)$ for its starting time and completion time, respectively.

\subparagraph{Event Times and Event Batches.}
We describe the algorithm in terms of event times. An event time is either time $0$ or a time at which at least one operation completes. Since operations are nonpreemptive and new execution operations are released only when tests complete, it suffices to specify the algorithm's behavior at event times. At an event time, the algorithm first marks all completed operations and releases the execution operation of every job whose test has just completed. It then assigns idle machines to available unscheduled operations. If fewer operations are available than idle machines, the remaining machines stay idle until the next event time.
The set of operations started at the same event time is called an event batch. Two schedules are batch-equivalent if they start the same event batch at every event time. In particular, batch-equivalent schedules have the same starting times, completion times, and objective value.

\subparagraph{Parallel 1-SORT.}
For an available operation $o$, define its running time 
%%yc
by
% as:
\[
    \ell(o):=
    \begin{cases}
        t_j, & \text{if } o=\tau_j,\\
        p_j, & \text{if } o=\pi_j.
    \end{cases}
\]
The value $\ell(\pi_j)=p_j$ is used only after $\pi_j$ has become available,
that is, after $\tau_j$ has completed.

\begin{center}
\fbox{
\begin{minipage}{0.92\linewidth}
\textbf{Parallel 1-SORT}

\smallskip
\textbf{Input:} A set of jobs on $m$ identical machines.

\smallskip
\textbf{Initialization:}
All tests $\tau_j$ are available at time $0$. No execution operation is
available initially.

\smallskip
\textbf{At each event time $t$:}
\begin{enumerate}
    \item Mark all operations completed at time $t$.
    \item For every job $j$ whose test $\tau_j$ has just completed, reveal
    $p_j$ and make $\pi_j$ available.
    \item Let $r$ be the number of idle machines.
    \item Among all available unscheduled operations, choose up to $r$ operations
    with smallest running times $\ell(o)$, breaking ties deterministically.
    \item Start these operations on the idle machines.
\end{enumerate}
\end{minipage}
}
\end{center}

\begin{theorem}
\label{thm:parallel-lifted-ratio}
Suppose that single-machine 1-SORT is $\rho$-competitive and $\rho \le 2$. Then parallel 1-SORT is
$
    \frac{2(m+\rho-1)}{m+1}
$
-competitive on $m$ identical machines.
\end{theorem}

%%yc

\subparagraph{Proof Setup.}
To prove Theorem~\ref{thm:parallel-lifted-ratio}, we compare parallel 1-SORT with the operation sequence obtained by applying single-machine 1-SORT to the same instance.

% \begin{proof}
% To analyze the algorithm, we compare it with the operation sequence obtained by
% applying single-machine 1-SORT to the same instance. 
Let
\[
    L=(o_1,o_2,\ldots,o_{2N})
\]
be this sequence. For two operations $a$ and $b$, we write $a\prec_L b$ if
$a$ appears before $b$ in $L$. Since the single-machine schedule is feasible, $\tau_j\prec_L\pi_j$ for every job $j$.
For two operations $a$ and $b$, we write $a\prec_P b$ if $a$ has a
smaller running time than $b$,
%%yc
with ties
% ties are
resolved by the fixed deterministic rule.
%%yc
We
% Finally, 
define the $m$-machine list schedule induced by $L$
%%yc add
as follows.
At each event time, after updating the available set, this schedule scans $L$ from left to
right and starts the first available unscheduled operations on the idle machines.
The first step of the analysis is to prove that this induced list schedule starts
exactly the same event batch as parallel 1-SORT at every event time.

%%yc
\subparagraph{Equivalence with the Induced List Schedule.}
% \subparagraph{Batch Equivalence.}
We first prove that parallel 1-SORT and the $m$-machine list schedule induced by
$L$ start the same event batch at every event time.

\begin{lemma}[Prefix-clean invariant]
\label{lem:prefix-clean}
Fix a test operation $\tau_x$. In the $m$-machine list schedule induced by $L$,
after $\tau_x$ has been scheduled, no operation $q\prec_L \tau_x$ can remain
available and unscheduled at the end of any event batch.
\end{lemma}

\begin{proof}
Consider the event batch in which $\tau_x$ is scheduled. If an available
unscheduled operation $q\prec_L\tau_x$ were left unscheduled in this batch, then
the list scheduler would encounter $q$ before $\tau_x$ and would schedule $q$
before $\tau_x$, a contradiction.

Now consider any later event batch, and assume the invariant holds before this batch starts. Any operation $q\prec_L\tau_x$ that becomes newly available in
this batch cannot be a test, since all tests are available from time $0$.
Thus $q$ is an execution operation whose test completes at this event time.
Each such newly available execution is accompanied by the completion of its own test, which frees one machine. When the list scheduler scans $L$, these newly
available operations before $\tau_x$ are encountered before any operation after
$\tau_x$. Hence they are scheduled in the same event batch, and none remains
available and unscheduled afterward.
\end{proof}

\begin{lemma}[Test monotonicity]
\label{lem:test-monotonicity}
If $\tau_a\prec_L\tau_b$, then in the $m$-machine list schedule induced by $L$,
\[
    S(\tau_a)\le S(\tau_b).
\]
\end{lemma}

\begin{proof}
Both tests are available from time $0$. If $\tau_b$ were scheduled before
$\tau_a$, then when $\tau_b$ is selected, the operation $\tau_a$ would still be available and unscheduled. Since $\tau_a\prec_L\tau_b$, the list scheduler would select $\tau_a$ before $\tau_b$, a contradiction.
\end{proof}

Since all tests are available from time $0$ and $\prec_P$ compares running times,
the single-machine sequence $L$ orders tests by nondecreasing test time.  Hence,
if $\tau_a\prec_L\tau_b$, then $t_a\le t_b$. Together with
Lemma~\ref{lem:test-monotonicity}, this implies
$
    C(\tau_a)\le C(\tau_b)
    % \qquad
    \text{ whenever } \tau_a\prec_L\tau_b.
$

%%yc
% \textcolor{blue}{
\begin{lemma}[Batch equivalence]
% \begin{theorem}[Batch equivalence]
\label{lem:batch-equivalence}
The parallel 1-SORT schedule and the $m$-machine list schedule induced by $L$
are batch-equivalent.
\end{lemma}
% }
% \end{theorem}

\begin{proof}
Assume for contradiction that the two schedules first choose different event
batches at some event time $t$. 
%%yc add
Before time $t$, they have chosen the same event batches. Hence,
immediately before the choices at time $t$,
%%yc delete
% the two schedules have the same history.  Hence 
they have the same available set
$A(t)$ and the same number $r$ of idle machines. Let $B_P(t)$ be the event batch chosen by parallel 1-SORT, and let $B_L(t)$ be
the event batch chosen by the list schedule. Choose $h\in B_P(t)\setminus B_L(t)$. Since $h$ is skipped by the list schedule, at least $r$ available operations
appear before $h$ in $L$.  
Define $E:=\{q\in A(t): q\prec_L h\}$. Then $|E|\ge r$.  
Since parallel 1-SORT chooses $h$,
%%yc ??
not all operations in $E$ can have priority higher than $h$ under \(\prec_P\). Hence
% while the list schedule chooses only operations before $h$, 
there exists an operation $o\in E$ such that
\[
    o\prec_L h
    \qquad\text{and}\qquad
    h\prec_P o.
\]

We first show that $h$ is an execution operation. If $h$ were a test, then $h$ would be available from time $0$. When the single-machine 1-SORT sequence chose $o$, the operation $h$ would also be available. 
Since $h\prec_P o$, 1-SORT
would choose $h$ before $o$, contradicting $o\prec_L h$.  
Hence $h=\pi_j$ for some job $j$.
Next we show that $o\prec_L\tau_j$.  
Suppose instead that $\tau_j\prec_L o$.  
Since $o\prec_L\pi_j$, when the single-machine sequence chose
$o$, the operation $\pi_j$ was already available and unscheduled. Since
$\pi_j=h\prec_P o$, 1-SORT would choose $\pi_j$ before $o$, contradicting
$o\prec_L\pi_j$. Therefore
\[
    o\prec_L\tau_j\prec_L\pi_j=h.
\]
Moreover, since $\tau_j$ is a test and is available from time $0$, the fact that
the single-machine sequence chooses $o$ before $\tau_j$ implies $o\prec_P\tau_j$.
Together with $\pi_j=h\prec_P o$, we obtain
$\pi_j\prec_P o\prec_P\tau_j$.
In particular,
$\pi_j\prec_P\tau_j$.

We now prove that $h=\pi_j$ becomes newly available at time $t$.  Since
$h\in A(t)$, the test $\tau_j$ has completed by time $t$.  Since
$o\prec_L\tau_j$ and $o\in A(t)$, Lemma~\ref{lem:prefix-clean} implies that
$o$ cannot be old available backlog after $\tau_j$ has been scheduled.  Hence
$o$ becomes newly available at time $t$.  Since tests are available from time
$0$, $o$ is an execution operation. Write $o=\pi_a$.  Then
\[
    \tau_a\prec_L\pi_a=o\prec_L\tau_j.
\]
Thus $\tau_a\prec_L\tau_j$, and therefore $C(\tau_a)\le C(\tau_j)$.
Because $o=\pi_a$ becomes newly available at time $t$, we have
$C(\tau_a)=t$.  Since $\pi_j$ is available at time $t$, we also have
$C(\tau_j)\le t$.  Hence $C(\tau_j)=t$,
so $h=\pi_j$ is newly available at time $t$. 

It remains to show that every operation in $E$ is newly available at time $t$. Let $q\in E$. 
% Consider the following two cases:
% \textbf{case 1}: 
If $q\prec_L\tau_j$, then Lemma~\ref{lem:prefix-clean} again implies that $q$
cannot be old available backlog after $\tau_j$ has been scheduled.  Since
$q\in A(t)$, it must be newly available at time $t$.
It remains to consider the case $\tau_j\prec_L q\prec_L\pi_j$.
We first show that $q$ cannot be a test.  Suppose $q=\tau_b$.  Since
$\tau_j\prec_L\tau_b$ and both tests are available from time $0$, we have $\tau_j\prec_P\tau_b$.
Also, when the single-machine sequence chose $\tau_b$, the execution
$\pi_j$ was already available and unscheduled.  Since $\tau_b\prec_L\pi_j$,
we have $\tau_b\prec_P\pi_j$.
Together with $\pi_j\prec_P\tau_j$, this gives the cycle
$\pi_j\prec_P\tau_j\prec_P\tau_b\prec_P\pi_j$,
%%yc add
which is
impossible.  Hence $q$ is an execution operation. 
Write $q=\pi_b$. We claim that $\tau_j\prec_L\tau_b$.  Suppose instead that
$\tau_b\prec_L\tau_j$. When the single-machine sequence chose $\tau_j$, the
operation $\pi_b$ was already available and unscheduled.  Since
$\tau_j\prec_L\pi_b$, we have $\tau_j\prec_P\pi_b$.
When the single-machine sequence chose $\pi_b$, the operation $\pi_j$ was already
available and unscheduled.  Since $\pi_b\prec_L\pi_j$, we have $\pi_b\prec_P\pi_j$.
Together with $\pi_j\prec_P\tau_j$, this gives the cycle $\pi_j\prec_P\tau_j\prec_P\pi_b\prec_P\pi_j$, again impossible. Therefore $\tau_j\prec_L\tau_b\prec_L\pi_b=q\prec_L\pi_j$. By the test-completion monotonicity proved above, $C(\tau_j)\le C(\tau_b)$.
We already know that $C(\tau_j)=t$.  Since $q=\pi_b$ is available at time $t$,
we also have $C(\tau_b)\le t$.  Therefore $C(\tau_b)=t$, so $q=\pi_b$ is newly
available at time $t$.
Thus every operation in $E$ is newly available at time $t$, and $h$ is newly
available at time $t$ as well.  These are distinct execution operations, so they
come from distinct tests completing at time $t$.  Each such test completion frees
one machine.  Hence at least $|E|+1$ machines are idle at time $t$, so $r\ge |E|+1$.
But $|E|\ge r$, giving $r\ge |E|+1>|E|\ge r$,
a contradiction.  Therefore no first divergence exists, and the two schedules
are batch-equivalent.
\end{proof}

\subparagraph{Completion-time Bound.}
We now analyze the $m$-machine list schedule induced by $L$.  By Lemma~\ref{lem:batch-equivalence}, the same bounds apply to parallel 1-SORT. 
% The argument below uses prefix-volume bounds for list schedules; it does not assume that all machines are continuously busy before a given start time.

Fix a job $j$.  Define
\[
    U_j:=\sum_{q\prec_L\tau_j}\ell(q),
    \qquad
    V_j:=\sum_{\tau_j\prec_L q\prec_L\pi_j}\ell(q).
\]
Thus $U_j$ is the total running time before $\tau_j$ in the single-machine
sequence, and $V_j$ is the total running time strictly between $\tau_j$ and
$\pi_j$.  Therefore the completion time of job $j$ in the single-machine
1-SORT schedule is
$
    C_j^{(1)}
    =
    U_j+t_j+V_j+p_j
    =
    U_j+V_j+\sigma_j.
$

\begin{lemma}
\label{lem:parallel-completion-bound}
In the $m$-machine list schedule induced by $L$,
\[
    C_j^{(m)}
    \le
    \frac{U_j+V_j}{m}+\sigma_j.
\]
\end{lemma}

\begin{proof}
Since $\tau_j$ is available from time $0$, the list scheduler can postpone
$\tau_j$ only by scheduling operations before $\tau_j$ in $L$.  The total running
time of these operations is $U_j$.  Hence $S(\tau_j)\le \frac{U_j}{m}$, and therefore $C(\tau_j)\le \frac{U_j}{m}+t_j$. Next consider $\pi_j$.  This operation cannot start before $\tau_j$ completes.
The total running time of operations preceding $\pi_j$ in $L$ is
\[
    U_j+t_j+V_j.
\]
By the list-scheduling rule, once $\pi_j$ is available, an idle machine cannot
be assigned to an operation after $\pi_j$ in $L$ while $\pi_j$ remains
unscheduled.  Therefore the standard prefix-volume bound for list schedules
gives
$
    S(\pi_j)
    \le
    \max\left\{
        C(\tau_j),\,
        \frac{U_j+t_j+V_j}{m}
    \right\}.
$
Since
$
    C(\tau_j)\le \frac{U_j}{m}+t_j,
$
we obtain
\[
    S(\pi_j)
    \le
    \max\left\{
        \frac{U_j}{m}+t_j,\,
        \frac{U_j+t_j+V_j}{m}
    \right\}.
\]
Adding $p_j$ gives
\[
    C_j^{(m)}
    \le
    \max\left\{
        \frac{U_j}{m}+t_j+p_j,\,
        \frac{U_j+t_j+V_j}{m}+p_j
    \right\}.
\]
Both terms inside the maximum are at most
$\frac{U_j+V_j}{m}+t_j+p_j
    =
    \frac{U_j+V_j}{m}+\sigma_j$.
This proves the lemma.
\end{proof}

%%yc add??
% \textcolor{blue}{
\begin{proof}[Proof of Theorem~\ref{thm:parallel-lifted-ratio}]
By Lemma~\ref{lem:batch-equivalence}, parallel 1-SORT has the same completion times as the $m$-machine list schedule induced by $L$. By Lemma~\ref{lem:parallel-completion-bound}, for every job $j$,
\[
    C_j^{(m)}
    \le
    \frac{U_j+V_j}{m}+\sigma_j.
\]
Since $C_j^{(1)}=U_j+V_j+\sigma_j$,
%%yc
this 
% Lemma~\ref{lem:parallel-completion-bound} 
implies
\[
    C_j^{(m)}
    \le
    \frac{1}{m}C_j^{(1)}
    +
    \left(1-\frac{1}{m}\right)\sigma_j.
\]
Summing over all jobs, we obtain
\[
    \ALG_m
    \le
    \frac{1}{m}\ALG_1
    +
    \left(1-\frac{1}{m}\right)\sum_j\sigma_j.
\]
If single-machine 1-SORT is $\rho$-competitive, then $\ALG_1\le \rho\,\OPT_1$.
Thus
\[
    \ALG_m
    \le
    \frac{\rho}{m}\OPT_1
    +
    \left(1-\frac{1}{m}\right)\sum_j\sigma_j.
\]
% \subparagraph{Competitive ratio.}
% It remains to compare the above upper bound with lower bounds on the optimum.
Order the jobs so that
\[
    \sigma_1\ge \sigma_2\ge\cdots\ge \sigma_N,
\]
and define
\[
    A:=\sum_{j=1}^{N} j\sigma_j,
    \qquad
    B:=\sum_{j=1}^{N}\sigma_j.
\]
For one machine, $\OPT_1=A$.
For $m$ identical machines, we use the standard lower bound from~\cite{GongChenHayashi2024},
\[
    \OPT_m
    \ge
    \frac{1}{m}A
    +
    \frac{1}{2}\left(1-\frac{1}{m}\right)B.
\]
%%yc
Combining the two bounds yields
% Combining this lower bound with the algorithmic upper bound gives
\[
    \frac{\ALG_m}{\OPT_m}
    \le
    \frac{
        \frac{\rho}{m}A+\left(1-\frac{1}{m}\right)B
    }{
        \frac{1}{m}A+\frac{1}{2}\left(1-\frac{1}{m}\right)B
    }
    =
    \frac{\rho A+(m-1)B}{A+\frac{m-1}{2}B}.
\]
Let $x:=\frac{A}{B}$.
Since $A=\sum_{j=1}^{N}j\sigma_j$ and $B=\sum_{j=1}^{N}\sigma_j$, we have
$x\ge 1$.  Define
\[
    R(x):=\frac{\rho x+m-1}{x+\frac{m-1}{2}}.
\]
If $\rho\le2$, then
\[
    R'(x)
    =
    \frac{(m-1)(\rho/2-1)}
    {\left(x+\frac{m-1}{2}\right)^2}
    \le 0.
\]
Thus 
%%yc add
the function
$R(x)$ is nonincreasing on $[1,\infty)$. Hence
\[
    R(x)\le R(1)
    =
    \frac{\rho+m-1}{1+\frac{m-1}{2}}
    =
    \frac{2(m+\rho-1)}{m+1}.
\]
This proves the theorem.
\end{proof}

\section{Conclusion}
\label{sec:con}

We study obligatory-test scheduling on $m$ identical machines with the objective of minimizing the sum of completion times. The main contribution is a completion-threshold 
%%yc
% \textcolor{blue}{
perspective for lower-bounds, which measures the global progress of an online algorithm through the times at which successive numbers of jobs are completed.
%}
% lower-bound framework that reasons directly about the global rate at which an online algorithm can reveal and complete jobs under total machine capacity. 
This viewpoint avoids the need for single-machine exchange arguments and is well suited to the parallel setting, where tests and processing operations may be interleaved across machines. The resulting lower bounds show that the difficulty of obligatory testing persists under parallelism. In particular, the dyadic construction gives a lower bound tending to $3/2$, and it also raises the deterministic single-machine lower bound from $\sqrt{2}$ to $3/2$.
% Using this framework, we prove a three-type lower bound of $1.4811$ and a dyadic multi-type lower bound tending to $3/2$.  The dyadic construction also improves the deterministic single-machine lower bound from $\sqrt{2}$ to $3/2$.  
On the algorithmic side, 
%%yc
the lifting theorem shows that guarantees for single-machine 1-SORT can be transferred to identical parallel machines, yielding a competitive ratio of $\frac{2(m+\rho-1)}{m+1}$ whenever the single-machine ratio is $\rho$.
% we proved a lifting theorem for 1-SORT: if single-machine 1-SORT is $\rho$-competitive, then parallel 1-SORT is $\frac{2(m+\rho-1)}{m+1}$-competitive on $m$ identical machines. 

Several questions remain open.  The most immediate one is to close the gap between the $3/2$ lower bound and the best known upper bounds for obligatory testing. Another direction is to design sharper algorithms for the multi-machine model, especially under uniform test times.  Finally, it would be interesting to develop lower-bound techniques that go beyond 
%%yc
completion-threshold arguments
% threshold counting 
and capture more detailed interactions between testing and processing on parallel machines.
% \nolinenumbers

% \bibliography{main-arxiv}

\end{document}